\documentclass[11pt,a4paper]{article}
\usepackage[T1]{fontenc}
\usepackage[margin=2.4cm]{geometry}
\usepackage{amsmath,amssymb,amsthm}
\usepackage{graphicx}
\graphicspath{{figures/}{./}}
\usepackage[numbers,sort&compress]{natbib}
\usepackage{booktabs}
\usepackage{hyperref}
\hypersetup{colorlinks=true,linkcolor=blue,citecolor=blue,urlcolor=blue}

\numberwithin{equation}{section}
\newtheorem{lemma}{Lemma}[section]
\newtheorem{proposition}[lemma]{Proposition}
\newtheorem{theorem}[lemma]{Theorem}
\newtheorem{corollary}[lemma]{Corollary}
\newcommand{\bcdot}{\boldsymbol{\cdot}}
\newcommand{\bnabla}{\boldsymbol{\nabla}}
\newcommand{\PT}{\mathbb{P}}
\newcommand{\Uc}{\boldsymbol{\mathcal{U}}}
\title{Exact Periodic Solutions of the Forced Incompressible\\
Navier--Stokes Equations in Arbitrary Dimensions}

\author{R. K. Michael Thambynayagam\\[2pt]
\normalsize Managing Director (Retired), Schlumberger Cambridge Research
\ \raisebox{-1pt}{\includegraphics[width=0.45cm]{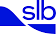}}}
\date{}

\begin{document}
\maketitle

\begin{abstract}
Exact periodic solutions of the unforced incompressible Navier--Stokes equations require the
convective field to be a pure gradient, absorbed into the pressure; for the cyclic trigonometric
families studied previously this occurs only at isolated phases and, over the range so far
classified, only at $n=3$ and $4$. The forced problem, we show, admits an exact construction for
every phase vector and every $n\ge3$. For an $n$-parameter family of two-term cyclic fields,
solenoidal and a Laplacian eigenfunction, the transverse part $\Uc^0$ of its convective field
serves as the body force, giving a closed-form solution with velocity, pressure and force explicit. We prove
closed forms for the pressure and for both magnitudes, in every dimension and with no upper bound
on $n$: the forcing is never trivial, and the $L^2$ magnitude of $\Uc^0$ stands to that of the
whole convective field in the constant ratio $2\sqrt2/3$, whatever the dimension, phase or units.
The forcing does no net work, so the kinetic energy decays purely viscously at arbitrary Reynolds
number, giving an exact benchmark in any dimension. A numerical study of vorticity concentration at
zero applied work uses it as a reference, and shows that apparent saturation can be a resolution
artefact.
\end{abstract}

\medskip
\noindent\textbf{Subject Areas:} applied mathematics, fluid mechanics, mathematical physics



\medskip
\protect{\bfseries{ORCID:}} https://orcid.org/0000-0002-1778-7327

\medskip
\protect{\noindent\textbf{Author for correspondence:} }R. K. Michael Thambynayagam\\
\texttt{michael.thambynayagam@gmail.com}

\section{Introduction}
\label{sec:intro}

Closed-form solutions of the incompressible Navier--Stokes equations are rare, and known examples
are typically obtained by arranging for the nonlinear term to vanish from the velocity dynamics. If
the convective field $\boldsymbol{g}=(\boldsymbol{v}\bcdot\bnabla)\boldsymbol{v}$ is a pure
gradient it can be absorbed entirely into the pressure, leaving the velocity to satisfy a linear
diffusion equation. The Taylor vortex \citep{taylor1923} is the classical instance, and the
periodic families of \citet{ant2020} and \citet{tha2023,tha2026} are constructed on the
same principle.

That principle is demanding. In \citet{tha2026} the condition is imposed on an $n$-parameter family
of cyclic trigonometric fields and the admissible phase vectors are classified for $3\le n\le8$:
they exist at $n=3$ and $n=4$, and not for $5\le n\le8$, with the obstructions dividing according
to the parity of the dimension. Generic phase vectors therefore yield no unforced solution at all.

This paper observes that the same phase space is entirely productive for the \emph{forced} problem,
and that the object responsible for the obstruction is precisely the object that supplies the
forcing: the same transverse component that obstructs the unforced construction supplies exactly
the forcing the forced solution requires. The result needs no iteration, no smallness of the data,
and no condition on the phases.

\subsection{Statement of the problem}
\label{sec:problem}

Throughout, the flow occupies the periodic cell
\begin{equation}
\label{eq:cell}
\Omega=[0,L]^n,
\qquad
\langle h\rangle=\frac{1}{\lvert\Omega\rvert}\int_\Omega h\,\mathrm{d}\boldsymbol{x},
\qquad
E(t)=\tfrac12\bigl\langle\lvert\boldsymbol{v}\rvert^2\bigr\rangle ,
\end{equation}
so that $\langle\cdot\rangle$ is the cell average and $E$ the kinetic energy per unit mass. All
fields are smooth and $L$-periodic in each coordinate, $\rho$ is the constant density and $\kappa$
the kinematic viscosity. The body force $\boldsymbol{f}$ is taken per unit mass, so that it has the
dimensions of acceleration, consistently with the term $-\bnabla p/\rho$ beside it. On $\Omega$,
\begin{equation}
\label{eq:ns}
\frac{\partial\boldsymbol{v}}{\partial t}
= \kappa\Delta\boldsymbol{v} - \bnabla p/\rho - \boldsymbol{g} + \boldsymbol{f},
\qquad
\bnabla\bcdot\boldsymbol{v}=0,
\qquad
\boldsymbol{g}=(\boldsymbol{v}\bcdot\bnabla)\boldsymbol{v}.
\end{equation}
\citet{tha2023} recast \eqref{eq:ns} into viscous, inertial and applied contributions,
\begin{equation}
\label{eq:proj7}
\frac{\partial v_i}{\partial t} = \kappa\Delta v_i - \mathcal{U}_i + \mathcal{F}_i ,
\end{equation}
in which the last two are given by the Newtonian integrals
\begin{align}
\label{eq:U}
\mathcal{U}_i(\boldsymbol{x},t) &= g_i(\boldsymbol{x},t)
 - \frac{\Gamma(n/2)}{2\pi^{n/2}}\int_{\mathbb{R}^n}
   \frac{(x_i-y_i)\sum_{k=1}^{n}\partial g_k(\boldsymbol{y},t)/\partial y_k}
        {\lvert\boldsymbol{x}-\boldsymbol{y}\rvert^{\,n}}\,\mathrm{d}\boldsymbol{y},\\
\label{eq:F}
\mathcal{F}_i(\boldsymbol{x},t) &= f_i(\boldsymbol{x},t)
 - \frac{\Gamma(n/2)}{2\pi^{n/2}}\int_{\mathbb{R}^n}
   \frac{(x_i-y_i)\sum_{k=1}^{n}\partial f_k(\boldsymbol{y},t)/\partial y_k}
        {\lvert\boldsymbol{x}-\boldsymbol{y}\rvert^{\,n}}\,\mathrm{d}\boldsymbol{y}.
\end{align}
In each, the integral is the longitudinal part of the field recovered from its divergence, so that
$\mathcal{U}_i$ and $\mathcal{F}_i$ are what remains after that part is removed. Writing $\PT$ for
the Leray projection \citep{leray1934} onto divergence-free fields, this says simply
\begin{equation}
\label{eq:UF}
\Uc = \PT\boldsymbol{g},
\qquad
\boldsymbol{\mathcal{F}} = \PT\boldsymbol{f}.
\end{equation}

Equations \eqref{eq:U}--\eqref{eq:F} and \eqref{eq:UF} are one operator seen from two sides, and
the correspondence is worth setting down: it fixes the notation inherited from
\citet{tha2023,tha2026}, and the two forms suit different domains. The kernel in
\eqref{eq:U}--\eqref{eq:F} is built from the whole-space Green's function of the Laplacian, and the
representation it gives is the one appropriate to fields decaying at infinity, which is the setting
of the earlier work. Periodic fields do not decay, so on $\Omega$ these integrals are not invoked as
absolutely convergent Newtonian integrals. Their role is taken, as in \citet{tha2026}, by the
equivalent Fourier representation of the Helmholtz projection on the periodic cell, understood mode
by mode: with
$\boldsymbol{w}(\boldsymbol{x})=\sum_{\boldsymbol{k}}\widehat{\boldsymbol{w}}
(\boldsymbol{k})e^{\mathrm{i}\boldsymbol{k}\bcdot\boldsymbol{x}}$ over the dual lattice
$\boldsymbol{k}\in(2\pi/L)\mathbb{Z}^n$,
\begin{equation}
\label{eq:leray}
\widehat{\PT\boldsymbol{w}}(\boldsymbol{k})
=\left(I-\frac{\boldsymbol{k}\boldsymbol{k}^{\mathsf T}}{\lvert\boldsymbol{k}\rvert^{2}}\right)
\widehat{\boldsymbol{w}}(\boldsymbol{k}),
\qquad \boldsymbol{k}\neq\boldsymbol{0},
\qquad
\widehat{\PT\boldsymbol{w}}(\boldsymbol{0})=\widehat{\boldsymbol{w}}(\boldsymbol{0}).
\end{equation}
The zero mode is left untouched: a constant field is divergence-free and carries no pressure
gradient, so it belongs to the solenoidal part. For the convective field the point is immaterial in
any case, since incompressibility gives
$g_i=v_j\,\partial v_i/\partial x_j=\partial(v_iv_j)/\partial x_j$, and the cell mean of a
derivative of a periodic function vanishes, so $\langle\boldsymbol{g}\rangle=\boldsymbol{0}$
identically. Equation \eqref{eq:leray} is the definition of $\PT$ used throughout this paper;
\eqref{eq:U}--\eqref{eq:F} record where it comes from, and allow the reader to pass between the
notation used here and that of \citet{tha2023,tha2026} without re-deriving the correspondence.

The reformulation shows that when the inertial contribution vanishes,
\begin{equation}
\label{eq:cond}
\Uc \;\equiv\; \boldsymbol{0},
\end{equation}
the velocity equation collapses to the Cauchy diffusion equation and the pressure is recovered
from the pressure Poisson equation \eqref{eq:ppe}, equivalently from the longitudinal component of
$\boldsymbol{g}-\boldsymbol{f}$. The remaining, longitudinal parts of $\boldsymbol{g}$ and $\boldsymbol{f}$ are
absorbed into the pressure and do not act on the velocity; this is why only \eqref{eq:UF} enters
what follows. We import \eqref{eq:UF} and \eqref{eq:cond} from \citet{tha2023} and
\citet{tha2026}, where \eqref{eq:cond} is stated equivalently as a Newtonian integral identity and
its solutions are classified for $3\le n\le8$.

The question addressed here is what can be said when \eqref{eq:cond} fails. The answer is that the
failure is itself constructive, and uniformly so.

The paper has two objectives, and they are of different kinds. The first is analytical: to
establish an exact forced solution and to evaluate it in closed form, so that it can serve as a
benchmark in any dimension. Sections~\ref{sec:proj}--\ref{sec:dims} and \ref{sec:bench} carry this
out. The second is computational: to use that exact solution as a controlled reference for a
numerical study of vorticity concentration in which the energy injected by the force is held at
zero exactly. Section~\ref{sec:conc} reports it. The two are connected---the second is possible
only because the first identifies the transverse part explicitly and shows that balancing it does
no work---but they are logically independent, and a reader interested only in the exact solution may
stop at the end of Section~\ref{sec:bench}.

It is worth separating two levels of generality at the outset. Theorem~\ref{thm:main} is not tied
to any particular field: it holds for \emph{any} smooth, solenoidal, periodic Laplacian
eigenfunction, in any dimension. The explicit formulas---the velocity family \eqref{eq:family}, its
pressure, its body force and their magnitudes---belong to one convenient choice of such a field.
Neither level carries a restriction on the dimension. Section~\ref{sec:dims} proves the closed forms
for that family for every $n\ge3$, with no upper bound: the cyclic construction has a finite
interaction range, so each component of the convective field involves five coordinates whatever $n$
may be, and the quantities computed from it are therefore the same functions in every dimension.
This is the point at which the present paper differs most sharply from \citet{tha2026}, where the
admissible phases had to be determined dimension by dimension and the classification stops at the
largest case examined.

\section{The projected equations}
\label{sec:proj}

Taking the divergence of \eqref{eq:ns} gives the pressure Poisson equation
\begin{equation}
\label{eq:ppe}
\Delta p = \rho\,\bnabla\bcdot(\boldsymbol{f}-\boldsymbol{g}).
\end{equation}
Any smooth periodic field splits uniquely into a gradient and a divergence-free remainder,
$\boldsymbol{w}=\bnabla\phi+\boldsymbol{h}$ with $\bnabla\bcdot\boldsymbol{h}=0$: the Helmholtz
decomposition. On the torus the uniqueness requires the usual normalisation of the potential,
$\langle\phi\rangle=0$, and the constant Fourier mode belongs to $\boldsymbol{h}$, in accordance
with \eqref{eq:leray}. In Fourier variables, each mode is associated with a wavevector
$\boldsymbol{k}=(k_1,\dots,k_n)$, which specifies its spatial frequencies. The longitudinal part of
the mode is its component parallel to $\boldsymbol{k}$, while $\PT$ returns the component
perpendicular to $\boldsymbol{k}$. Comparing with \eqref{eq:ppe}, the pressure
gradient accounts for the longitudinal part of $\boldsymbol{g}-\boldsymbol{f}$, and \eqref{eq:ns}
is equivalent to
\begin{equation}
\label{eq:proj}
\frac{\partial\boldsymbol{v}}{\partial t}
= \kappa\Delta\boldsymbol{v} - \PT\boldsymbol{g} + \PT\boldsymbol{f},
\end{equation}
with the pressure eliminated. Condition \eqref{eq:cond} is now transparent, and so is the
construction of Section~\ref{sec:main}: if $\boldsymbol{f}$ is chosen so that
$\PT\boldsymbol{f}=\PT\boldsymbol{g}$, the two terms cancel and the velocity again satisfies pure
diffusion---this time without any condition on the field.

\section{The construction in $n$ dimensions}
\label{sec:family}

We need a family of initial fields that is defined in every dimension, is solenoidal without
having to solve anything, and is simple enough that $\Uc$ can be computed in closed form. The
first requirement rules out the stream-function representation, which uses the vector-valued curl
and does not extend beyond three dimensions. The second is met by a cyclic construction: if each
component is built from the same pattern of factors, shifted by one coordinate, the terms of
$\bnabla\bcdot\boldsymbol{v}^0$ can be made to cancel in pairs under a relabelling of the summation
index, as Lemma~\ref{lem:sol} shows. The third suggests using as few factors as possible, so that
the convective field stays within reach of hand computation; within the cyclic constructions
considered here, two factors per term is the smallest number that produces a non-trivial $\Uc$.

These considerations lead to the following. Let $\alpha=2\pi/L$ be the wavenumber, $L$ the period,
$v_r$ a reference velocity, and $\boldsymbol{\xi}=(\xi_1,\dots,\xi_n)$ a vector of phase angles.
Following \citet{tha2026}, subscripts on the coordinates are read cyclically: for any integer $m$,
\begin{equation}
\label{eq:cyc}
\langle m\rangle \;:=\; \bigl((m-1)\bmod n\bigr)+1 \;\in\;\{1,\dots,n\},
\end{equation}
so that $x_{\langle 0\rangle}=x_n$ and $x_{\langle n+1\rangle}=x_1$. With this convention, define
\begin{align}
\label{eq:family}
\frac{v_i^0(\boldsymbol{x},\boldsymbol{\xi})}{v_r}
&= \sin\!\left(\alpha x_{\langle i\rangle}+\xi_{\langle i\rangle}\right)
   \sin\!\left(\alpha x_{\langle i-1\rangle}+\xi_{\langle i-1\rangle}\right)
\nonumber\\
&\quad
 + \cos\!\left(\alpha x_{\langle i\rangle}+\xi_{\langle i\rangle}\right)
   \cos\!\left(\alpha x_{\langle i+1\rangle}+\xi_{\langle i+1\rangle}\right),
\qquad i=1,\dots,n .
\end{align}
The bracket in \eqref{eq:cyc} is what makes the definition well posed at the ends of the range:
at $i=1$ the first product calls for $x_{\langle 0\rangle}$, which is $x_n$, and at $i=n$ the
second calls for $x_{\langle n+1\rangle}$, which is $x_1$. At $n=4$, for example, \eqref{eq:family} gives
$v_1^0/v_r=\sin(\alpha x_1+\xi_1)\sin(\alpha x_4+\xi_4)+\cos(\alpha x_1+\xi_1)\cos(\alpha
x_2+\xi_2)$ for the first component and
$v_4^0/v_r=\sin(\alpha x_4+\xi_4)\sin(\alpha x_3+\xi_3)+\cos(\alpha x_4+\xi_4)\cos(\alpha
x_1+\xi_1)$ for the last, these being the two components in which the bracket acts non-trivially.

Each component is thus a sum of two products, the first coupling $x_i$ to its predecessor and the
second to its successor, so that the family is an $n$-parameter set of two-term cyclic
trigonometric products defined in every dimension. It plays here the role that the $n$-fold cyclic
ansatz plays in \citet{tha2026}, but it is a different family: there the aim was to satisfy
\eqref{eq:cond}, and the ansatz is a difference of two $n$-fold products; here the aim is to
violate it as simply as possible, and two factors suffice.

\begin{lemma}[Solenoidality]
\label{lem:sol}
The field \eqref{eq:family} is divergence-free for every $n\ge2$ and every $\boldsymbol{\xi}$.
\end{lemma}

\begin{proof}
Differentiating,
\[
\frac{1}{\alpha v_r}\frac{\partial v_i^0}{\partial x_i}
= \cos\!\left(\alpha x_{\langle i\rangle}+\xi_{\langle i\rangle}\right)
  \sin\!\left(\alpha x_{\langle i-1\rangle}+\xi_{\langle i-1\rangle}\right)
- \sin\!\left(\alpha x_{\langle i\rangle}+\xi_{\langle i\rangle}\right)
  \cos\!\left(\alpha x_{\langle i+1\rangle}+\xi_{\langle i+1\rangle}\right).
\]
Summing over $i$ and reindexing the first sum by $j=i-1$ carries it into
$\sum_j\sin(\alpha x_{\langle j\rangle}+\xi_{\langle j\rangle})
\cos(\alpha x_{\langle j+1\rangle}+\xi_{\langle j+1\rangle})$, which is the second sum term by term.
The two cancel identically.
\end{proof}

\begin{lemma}[Eigenfunction property]
\label{lem:eig}
The field \eqref{eq:family} satisfies $\Delta\boldsymbol{v}^0=-2\alpha^2\boldsymbol{v}^0$ for every
$n$ and every $\boldsymbol{\xi}$.
\end{lemma}

\begin{proof}
Each of the two products contains exactly two factors of wavenumber $\alpha$, in distinct
coordinates, so each is annihilated by $\Delta+2\alpha^2$; the result follows by linearity.
\end{proof}

Two features of \eqref{eq:family} should be noted, since both differ from the construction of
\citet{tha2026} and a reader coming from that work will expect otherwise.

First, each component of \eqref{eq:family} depends on three coordinates only---$x_{\langle
i-1\rangle}$, $x_{\langle i\rangle}$ and $x_{\langle i+1\rangle}$---and not on all $n$. This is a
marked contrast with the $n$-fold cyclic ansatz of \citet{tha2026}, in which every component
depends on every coordinate. At $n=3$ that distinction disappears, since three consecutive indices
exhaust the coordinates; for $n\ge4$ the present field couples each direction only to its two
cyclic neighbours. The two constructions remain different families even at $n=3$---one is a
difference of $n$-fold products, the other a sum of two---and it is only the
three-coordinates-versus-all-$n$ contrast that ceases to bite there. The field is nonetheless genuinely
$n$-dimensional: it is defined on the $n$-torus, all $n$ coordinates appear across the components,
all $n$ phases are free, and the divergence cancels only because the chain of couplings closes on
itself. What the shorter products buy is that $\boldsymbol{g}^0$ and $\Uc^0$ remain tractable in
closed form as $n$ grows, which is what makes the dimensional results of Section~\ref{sec:dims}
available at all.

Second, \eqref{eq:family} is not the only possible choice, and nothing in the paper depends on its
being so. Theorem~\ref{thm:main} requires only that $\boldsymbol{v}^0$ be smooth, solenoidal,
periodic and a Laplacian eigenfunction; any such field furnishes an exact forced solution, and we
have verified the theorem numerically for randomly generated fields on several distinct wavenumber
shells as well as for \eqref{eq:family}. Products involving all $n$ factors return one to the type of family considered in
\citet{tha2026}, which is constructed to satisfy \eqref{eq:cond} rather than to violate it;
within the cyclic constructions considered here, a single factor per term does not yield a
solenoidal field with non-trivial $\Uc^0$. Two factors is the smallest number for which we obtain
both, together with the closed forms of Section~\ref{sec:dims}, and that is the sense in which
\eqref{eq:family} is distinguished: it is convenient, not canonical.

Both lemmas hold for every $n\ge3$ and every phase vector, and have been independently verified
symbolically for $3\le n\le8$; that range, and the others quoted in Section~\ref{sec:repro}, are
the extents of the independent checks, not limitations on the results, which are proved as stated.
We note in particular that the eigenvalue in Lemma~\ref{lem:eig} is
$-2\alpha^2$, independent of $n$, so the diffusive timescale does not vary with dimension.

\section{The forced solution}
\label{sec:main}

\begin{theorem}
\label{thm:main}
Let $\boldsymbol{v}^0$ be any smooth solenoidal periodic field satisfying
$\Delta\boldsymbol{v}^0=-\lambda\boldsymbol{v}^0$, and set
$\boldsymbol{g}^0=(\boldsymbol{v}^0\bcdot\bnabla)\boldsymbol{v}^0$ and $\Uc^0=\PT\boldsymbol{g}^0$.
Then
\begin{equation}
\label{eq:soln}
\boldsymbol{v}(\boldsymbol{x},t) = \boldsymbol{v}^0(\boldsymbol{x})\,e^{-\lambda\kappa t},
\qquad
p(\boldsymbol{x},t) = p^0(\boldsymbol{x})\,e^{-2\lambda\kappa t},
\qquad
\boldsymbol{f}(\boldsymbol{x},t) = \Uc^0(\boldsymbol{x})\,e^{-2\lambda\kappa t},
\end{equation}
in which $p^0$ is the zero-mean solution of
$\Delta p^0=-\rho\,\bnabla\bcdot\boldsymbol{g}^0$, is an exact solution of
the forced system \eqref{eq:ns}, the pressure being determined up to an arbitrary additive function
of time. The forcing does no net work,
\begin{equation}
\label{eq:work}
\int_\Omega \boldsymbol{f}\bcdot\boldsymbol{v}\,\mathrm{d}\boldsymbol{x}=0
\qquad\text{for all }t\ge0,
\end{equation}
and the kinetic energy decays purely viscously,
\begin{equation}
\label{eq:Edecay}
E(t)=E(0)\,e^{-2\lambda\kappa t}.
\end{equation}
\end{theorem}

\begin{proof}
Since $\boldsymbol{v}^0$ is an eigenfunction, $\partial_t\boldsymbol{v}=-\lambda\kappa\boldsymbol{v}$
and $\kappa\Delta\boldsymbol{v}=-\lambda\kappa\boldsymbol{v}$, so the left-hand side of
\eqref{eq:proj} cancels the viscous term identically. The convective field is quadratic, so
$\PT\boldsymbol{g}(t)=\Uc^0e^{-2\lambda\kappa t}=\PT\boldsymbol{f}(t)$, and the remaining two terms
of \eqref{eq:proj} cancel. Hence \eqref{eq:proj} is satisfied, and with the pressure recovered from
\eqref{eq:ppe} so is \eqref{eq:ns}. The right-hand side of \eqref{eq:ppe} is likewise quadratic in
$\boldsymbol{v}$, so it carries the factor $e^{-2\lambda\kappa t}$; since the Laplacian acts only
in space, the same factor passes to the pressure, giving
$p=p^0e^{-2\lambda\kappa t}$ as stated.

Two separate points about divergences are worth keeping apart. Solenoidality of the velocity is
immediate: $\boldsymbol{v}(t)=\boldsymbol{v}^0e^{-\lambda\kappa t}$ is a time-dependent multiple of
$\boldsymbol{v}^0$, so $\bnabla\bcdot\boldsymbol{v}=0$ for all $t$ because
$\bnabla\bcdot\boldsymbol{v}^0=0$. That $\Uc^0$ is divergence-free is a different statement: it says
that the force chosen in \eqref{eq:soln} is purely transverse, so that
$\PT\boldsymbol{f}=\boldsymbol{f}$ and no part of it is silently absorbed into the pressure.

For \eqref{eq:work}, $\PT$ is self-adjoint and $\PT\boldsymbol{v}=\boldsymbol{v}$, so
$\langle\boldsymbol{f}\bcdot\boldsymbol{v}\rangle
=\langle\PT\boldsymbol{g}\bcdot\boldsymbol{v}\rangle
=\langle\boldsymbol{g}\bcdot\boldsymbol{v}\rangle$, and for solenoidal $\boldsymbol{v}$
\[
\langle\boldsymbol{v}\bcdot(\boldsymbol{v}\bcdot\bnabla)\boldsymbol{v}\rangle
=\tfrac12\langle\boldsymbol{v}\bcdot\bnabla\lvert\boldsymbol{v}\rvert^2\rangle
=-\tfrac12\langle(\bnabla\bcdot\boldsymbol{v})\lvert\boldsymbol{v}\rvert^2\rangle=0 .
\]
The integration by parts here produces no boundary term, because $\Omega$ is periodic and the
integrand $\tfrac12\lvert\boldsymbol{v}\rvert^2\boldsymbol{v}$ is periodic with it. The zero-work
property is therefore a consequence of just two things: periodicity, and incompressibility through
the self-adjointness of $\PT$ and the identity above.
The energy identity then reduces to
\begin{equation}
\label{eq:Ebal}
\frac{\mathrm{d}E}{\mathrm{d}t}=-2\kappa\lambda E ,
\end{equation}
which integrates to \eqref{eq:Edecay}.
\end{proof}

Note that $\boldsymbol{f}$ in \eqref{eq:soln} is the applied force itself, not its projection. It
is chosen to be purely transverse, $\PT\boldsymbol{f}=\boldsymbol{f}$, so that
$\boldsymbol{\mathcal{F}}=\boldsymbol{f}=\Uc$ and \eqref{eq:proj7} closes. Adding any gradient to
$\boldsymbol{f}$ would leave $\boldsymbol{\mathcal{F}}$, and hence the velocity, unchanged; it
would merely shift the pressure by that amount.

Three remarks are worth emphasising. First, the theorem imposes no restriction on the Reynolds
number, since no expansion is made and no convergence is required. Second, unlike the
classification of \citet{tha2026}, it places no condition on the phases. Third, the forcing is not
freely prescribed: it is precisely the component of the convective field that cannot be absorbed
into the pressure, so the obstruction to the unforced problem becomes the mechanism that solves the
forced one.

Applied to the family \eqref{eq:family}, for which $\lambda=2\alpha^2$ by Lemma~\ref{lem:eig},
\begin{align}
\label{eq:solnfam}
\boldsymbol{v}(\boldsymbol{x},t)&=\boldsymbol{v}^0(\boldsymbol{x},\boldsymbol{\xi})\,
e^{-2\alpha^2\kappa t},
\nonumber\\
p(\boldsymbol{x},t)&=p^0(\boldsymbol{x},\boldsymbol{\xi})\,
e^{-4\alpha^2\kappa t},
\\
\boldsymbol{f}(\boldsymbol{x},t)&=\Uc^0(\boldsymbol{x},\boldsymbol{\xi})\,
e^{-4\alpha^2\kappa t},
\nonumber
\end{align}
in every dimension $n\ge3$ and for every phase vector. The velocity therefore relaxes on the
timescale $(2\alpha^2\kappa)^{-1}$, while the pressure and the constructed body force, both
inherited from quadratic convective quantities, decay at twice that exponential rate. The body
force is not intrinsically quadratic in the velocity; it is so here because of the way it is
constructed.

\section{The three-dimensional case in closed form}
\label{sec:3D}

We now write out the case $n=3$ of \eqref{eq:family} in full, retaining the wavenumber $\alpha$ and
the reference velocity $v_r$, so that the dependence of each quantity on them remains visible.

The phases may be set to zero without any loss of generality. In \eqref{eq:family} each phase
$\xi_{\langle m\rangle}$ accompanies its own coordinate $x_{\langle m\rangle}$ and no other, so
that
\begin{equation}
\label{eq:shift}
\boldsymbol{v}^0(\boldsymbol{x},\boldsymbol{\xi})
=\boldsymbol{v}^0\!\left(\boldsymbol{x}+\boldsymbol{\xi}/\alpha,\boldsymbol{0}\right),
\qquad
\bigl(\boldsymbol{\xi}/\alpha\bigr)_m=\xi_m/\alpha .
\end{equation}
A phase vector is therefore nothing but a rigid translation of the cell, applied independently
along each coordinate axis. Since the Navier--Stokes equations are translation invariant and the
Leray projection commutes with translations, every quantity computed below---the convective field,
its divergence, the pressure and the body force---is the corresponding $\boldsymbol{\xi}=\boldsymbol{0}$
expression evaluated at the shifted argument. Solutions thus exist for every phase vector, as
Theorem~\ref{thm:main} asserts, and the general expressions are recovered from those below by
replacing $\alpha x_{\langle m\rangle}$ with $\alpha x_{\langle m\rangle}+\xi_{\langle m\rangle}$
throughout. This is in marked contrast to the unforced classification of \citet{tha2026}, where the
phases enter the construction at positions independent of the coordinate they multiply and are
consequently not removable in this way.

With $\boldsymbol{\xi}=\boldsymbol{0}$, then, the field is
\begin{align}
\label{eq:v03}
v_1^0 &= v_r\bigl(\sin\alpha x_1\sin\alpha x_3 + \cos\alpha x_1\cos\alpha x_2\bigr), \nonumber\\
v_2^0 &= v_r\bigl(\sin\alpha x_2\sin\alpha x_1 + \cos\alpha x_2\cos\alpha x_3\bigr), \\
v_3^0 &= v_r\bigl(\sin\alpha x_3\sin\alpha x_2 + \cos\alpha x_3\cos\alpha x_1\bigr), \nonumber
\end{align}
with convective field $\boldsymbol{g}^0=(\boldsymbol{v}^0\bcdot\bnabla)\boldsymbol{v}^0$,
\begin{align}
\label{eq:g03}
g_1^0 &= \alpha v_r^2\Bigl[-\bigl(\sin\alpha x_1\sin\alpha x_2+\cos\alpha x_2\cos\alpha x_3\bigr)
        \sin\alpha x_2\cos\alpha x_1
        \nonumber\\
      &\quad -\bigl(\sin\alpha x_1\sin\alpha x_3+\cos\alpha x_1\cos\alpha x_2\bigr)
              \bigl(\sin\alpha x_1\cos\alpha x_2-\sin\alpha x_3\cos\alpha x_1\bigr)
        \nonumber\\
      &\quad +\bigl(\sin\alpha x_2\sin\alpha x_3+\cos\alpha x_1\cos\alpha x_3\bigr)
        \sin\alpha x_1\cos\alpha x_3\Bigr],
\end{align}
the remaining components following by cyclic permutation. Its divergence is
\begin{equation}
\label{eq:divg3}
\bnabla\bcdot\boldsymbol{g}^0
= -\alpha^2v_r^2\bigl[\sin\alpha x_1\sin 2\alpha x_2\cos\alpha x_3
       + \sin 2\alpha x_1\sin\alpha x_3\cos\alpha x_2
       + \sin\alpha x_2\sin 2\alpha x_3\cos\alpha x_1\bigr].
\end{equation}
Every harmonic here carries $\lvert\boldsymbol{k}\rvert^2=6\alpha^2$, so the Poisson equation
inverts by a single division, the factor $\alpha^2$ cancels, and the pressure is
\begin{equation}
\label{eq:p3}
p^0 = -\frac{\rho v_r^2}{6}\bigl[
\sin\alpha x_1\sin 2\alpha x_2\cos\alpha x_3 + \sin 2\alpha x_1\sin\alpha x_3\cos\alpha x_2
+ \sin\alpha x_2\sin 2\alpha x_3\cos\alpha x_1\bigr],
\end{equation}
whose amplitude is independent of $\alpha$, as dimensional consistency requires: a pressure divided
by density has the dimensions of $v_r^2$. The field itself of course still depends on $\alpha$,
through the arguments of the trigonometric factors.

The body force is $\boldsymbol{f}=\Uc^0e^{-4\alpha^2\kappa t}$ with
$\Uc^0=\boldsymbol{g}^0+\bnabla p^0/\rho$, which evaluates to
\begingroup\small
\begin{align}
\label{eq:U3}
\mathcal{U}_1^0 ={}&\tfrac{\alpha v_r^2}{6}\bigl[
 \sin\alpha(-2x_1+x_2+x_3)+\sin\alpha(-x_1+x_2+2x_3)-\sin\alpha(-x_1+2x_2+x_3)\nonumber\\
&+\sin\alpha(x_1-2x_2+x_3)+\sin\alpha(x_1-x_2+2x_3)+\sin\alpha(x_1+x_2-2x_3)\nonumber\\
&-\sin\alpha(x_1+x_2+2x_3)-\sin\alpha(x_1+2x_2-x_3)-\sin\alpha(x_1+2x_2+x_3)\nonumber\\
&+\sin\alpha(2x_1-x_2+x_3)-\sin\alpha(2x_1+x_2-x_3)+\sin\alpha(2x_1+x_2+x_3)\bigr],\nonumber\\[3pt]
\mathcal{U}_2^0 ={}&\tfrac{\alpha v_r^2}{6}\bigl[
 \sin\alpha(-2x_1+x_2+x_3)-\sin\alpha(-x_1+x_2+2x_3)-\sin\alpha(-x_1+2x_2+x_3)\nonumber\\
&+\sin\alpha(x_1-2x_2+x_3)-\sin\alpha(x_1-x_2+2x_3)+\sin\alpha(x_1+x_2-2x_3)\nonumber\\
&-\sin\alpha(x_1+x_2+2x_3)+\sin\alpha(x_1+2x_2-x_3)+\sin\alpha(x_1+2x_2+x_3)\nonumber\\
&+\sin\alpha(2x_1-x_2+x_3)+\sin\alpha(2x_1+x_2-x_3)-\sin\alpha(2x_1+x_2+x_3)\bigr],\\[3pt]
\mathcal{U}_3^0 ={}&\tfrac{\alpha v_r^2}{6}\bigl[
 \sin\alpha(-2x_1+x_2+x_3)+\sin\alpha(-x_1+x_2+2x_3)+\sin\alpha(-x_1+2x_2+x_3)\nonumber\\
&+\sin\alpha(x_1-2x_2+x_3)-\sin\alpha(x_1-x_2+2x_3)+\sin\alpha(x_1+x_2-2x_3)\nonumber\\
&+\sin\alpha(x_1+x_2+2x_3)+\sin\alpha(x_1+2x_2-x_3)-\sin\alpha(x_1+2x_2+x_3)\nonumber\\
&-\sin\alpha(2x_1-x_2+x_3)-\sin\alpha(2x_1+x_2-x_3)-\sin\alpha(2x_1+x_2+x_3)\bigr].\nonumber
\end{align}
\endgroup
Every wavevector appearing in \eqref{eq:U3} is $\alpha$ times a permutation of
$(\pm1,\pm1,\pm2)$, so that $\lvert\boldsymbol{k}\rvert^2=6\alpha^2$ throughout, and $\bnabla\bcdot\Uc^0=0$ identically.

Collecting \eqref{eq:v03}, \eqref{eq:p3} and \eqref{eq:U3} and attaching the time factors of
\eqref{eq:solnfam}, the three-dimensional solution reads, in full,
\begin{equation}
\label{eq:soln3}
\begin{aligned}
v_1(\boldsymbol{x},t) &= v_r\bigl(\sin\alpha x_1\sin\alpha x_3
   + \cos\alpha x_1\cos\alpha x_2\bigr)\,e^{-2\alpha^2\kappa t},
\\
v_2(\boldsymbol{x},t) &= v_r\bigl(\sin\alpha x_2\sin\alpha x_1
   + \cos\alpha x_2\cos\alpha x_3\bigr)\,e^{-2\alpha^2\kappa t},
\\
v_3(\boldsymbol{x},t) &= v_r\bigl(\sin\alpha x_3\sin\alpha x_2
   + \cos\alpha x_3\cos\alpha x_1\bigr)\,e^{-2\alpha^2\kappa t},
\\[4pt]
p(\boldsymbol{x},t) &= -\frac{\rho v_r^2}{6}\bigl[
\sin\alpha x_1\sin 2\alpha x_2\cos\alpha x_3
+ \sin 2\alpha x_1\sin\alpha x_3\cos\alpha x_2
\\
&\qquad\qquad\quad
+ \sin\alpha x_2\sin 2\alpha x_3\cos\alpha x_1\bigr]\,e^{-4\alpha^2\kappa t},
\\[4pt]
\boldsymbol{f}(\boldsymbol{x},t) &= \Uc^0(\boldsymbol{x})\,e^{-4\alpha^2\kappa t},
\qquad\text{with }\Uc^0\text{ given by \eqref{eq:U3}},
\end{aligned}
\end{equation}
the kinetic energy decaying as $E(t)=E(0)e^{-4\alpha^2\kappa t}$ by \eqref{eq:Edecay}. Velocity,
pressure and body force are thus all explicit, in space and in time, for arbitrary $\alpha$, $v_r$,
$\rho$ and $\kappa$; the expressions at a general phase vector follow from \eqref{eq:shift}.

\section{Closed forms in arbitrary dimension}
\label{sec:dims}

Everything in Section~\ref{sec:3D} can be carried through for general $n$, and the results are not
merely verifiable case by case: they can be proved in closed form, for every $n\ge3$ and every
phase vector. The mechanism is that the cyclic construction has a finite interaction range. Each
component of $\boldsymbol{v}^0$ involves only three consecutive coordinates, so each component of
the convective field involves only five, whatever the dimension; and the three terms it reduces to
carry their doubled frequency in three different coordinates, which keeps them orthogonal.

Throughout this section $\lVert\cdot\rVert$ denotes the cell-averaged $L^2$ norm,
$\lVert\boldsymbol{w}\rVert=\langle\lvert\boldsymbol{w}\rvert^2\rangle^{1/2}$, which by Parseval's
identity equals the $\ell^2$ norm of the Fourier coefficients. That the quantities below do not
depend on the phases is immediate from \eqref{eq:shift}: a phase vector translates the field
rigidly, and a translation leaves an integral over the periodic cell unchanged. We nonetheless
carry $\boldsymbol{\xi}$ through the statements, writing
\begin{equation}
\label{eq:theta}
\theta_m \;:=\; \alpha x_{\langle m\rangle}+\xi_{\langle m\rangle},
\end{equation}
so that the family \eqref{eq:family} reads
$v_i^0=v_r(\sin\theta_i\sin\theta_{i-1}+\cos\theta_i\cos\theta_{i+1})$.

\begin{lemma}[The convective field in closed form]
\label{lem:gclosed}
For every $n\ge3$, every $\boldsymbol{\xi}$ and every $i$,
\begin{equation}
\label{eq:gclosed}
g_i^0 = \alpha v_r^2\Bigl[
\tfrac12\sin\theta_{i-2}\,\sin2\theta_{i-1}\,\sin\theta_i
\;+\;\sin\theta_{i-1}\,\cos2\theta_i\,\cos\theta_{i+1}
\;-\;\tfrac12\sin2\theta_{i+1}\,\cos\theta_i\,\cos\theta_{i+2}\Bigr].
\end{equation}
\end{lemma}

\begin{proof}
By \eqref{eq:family}, $v_i^0$ depends only on $x_{\langle i-1\rangle}$, $x_{\langle i\rangle}$ and
$x_{\langle i+1\rangle}$, so only $j=i-1,i,i+1$ contribute to
$g_i^0=\sum_j v_j^0\,\partial v_i^0/\partial x_j$, with
\begin{align*}
\frac{\partial v_i^0}{\partial x_{\langle i-1\rangle}}&=\alpha v_r\sin\theta_i\cos\theta_{i-1},
&
\frac{\partial v_i^0}{\partial x_{\langle i+1\rangle}}&=-\alpha v_r\cos\theta_i\sin\theta_{i+1},
\\
\frac{\partial v_i^0}{\partial x_{\langle i\rangle}}
 &=\alpha v_r\bigl(\cos\theta_i\sin\theta_{i-1}-\sin\theta_i\cos\theta_{i+1}\bigr). &&
\end{align*}
Multiplying by $v_{i-1}^0$, $v_i^0$ and $v_{i+1}^0$ respectively and expanding, the terms carrying
the factor $\sin\theta_i\cos\theta_i$ are
\[
\sin\theta_i\cos\theta_i\bigl(\cos^2\theta_{i-1}+\sin^2\theta_{i-1}\bigr)
-\sin\theta_i\cos\theta_i\bigl(\cos^2\theta_{i+1}+\sin^2\theta_{i+1}\bigr),
\]
which vanishes identically. Of what remains, the two terms in $\sin\theta_{i-1}\cos\theta_{i+1}$
combine as $\cos^2\theta_i-\sin^2\theta_i=\cos2\theta_i$, and the double-angle identity applied to
$\sin\theta_{i-1}\cos\theta_{i-1}$ and to $\sin\theta_{i+1}\cos\theta_{i+1}$ gives \eqref{eq:gclosed}.
\end{proof}

Three features of \eqref{eq:gclosed} do the work that follows. First, it involves only the five
coordinates $x_{\langle i-2\rangle},\dots,x_{\langle i+2\rangle}$, whatever $n$ may be. Second, each
of its three terms is a product of three factors in three \emph{distinct} coordinates---for the
first term these are $\langle i-2\rangle,\langle i-1\rangle,\langle i\rangle$, and consecutive
indices are distinct modulo $n$ for every $n\ge3$---one factor carrying frequency $2\alpha$ and two
carrying $\alpha$. Expanding each product gives eight harmonics, one for each choice of sign, and
all eight share the same vector of coordinate-wise magnitudes
$\bigl(\lvert k_1\rvert,\dots,\lvert k_n\rvert\bigr)$. In particular
\begin{equation}
\label{eq:k6}
\lvert\boldsymbol{k}\rvert^2=(2\alpha)^2+\alpha^2+\alpha^2=6\alpha^2
\end{equation}
for every mode present in \eqref{eq:gclosed}.

Third, the three terms have disjoint Fourier supports. Label them $A$, $B$ and $C$ in the order
written in \eqref{eq:gclosed}. Their magnitude vectors are
\begin{equation}
\label{eq:supports}
\begin{array}{lccccc}
 & \lvert k_{\langle i-2\rangle}\rvert & \lvert k_{\langle i-1\rangle}\rvert
 & \lvert k_{\langle i\rangle}\rvert & \lvert k_{\langle i+1\rangle}\rvert
 & \lvert k_{\langle i+2\rangle}\rvert \\[2pt]
A: & \alpha & 2\alpha & \alpha & & \\
B: &  & \alpha & 2\alpha & \alpha & \\
C: &  &  & \alpha & 2\alpha & \alpha .
\end{array}
\end{equation}
Compare $A$ with $B$ at the single coordinate $x_{\langle i-1\rangle}$: every harmonic of $A$ has
$\lvert k_{\langle i-1\rangle}\rvert=2\alpha$ there, and every harmonic of $B$ has
$\lvert k_{\langle i-1\rangle}\rvert=\alpha$. No wavevector can satisfy both, whatever the signs, so
the supports are disjoint. The same coordinate separates $A$ from $C$: at $n=3$ the columns
$\langle i+2\rangle$ and $\langle i-1\rangle$ coincide, so $C$ has
$\lvert k_{\langle i-1\rangle}\rvert=\alpha$ there, while for $n\ge4$ it has $0$; neither is
$2\alpha$. And $x_{\langle i\rangle}$ separates $B$ from $C$, with magnitudes $2\alpha$ and
$\alpha$. Each comparison is between magnitudes at one named coordinate, so it is unaffected both by
the choice of signs and by any coincidence among the cyclic indices, and it therefore holds for
every $n\ge3$. The three terms are mutually orthogonal in $L^2(\Omega)$.

\begin{corollary}[Divergence and pressure]
\label{cor:pn}
For every $n\ge3$ and every $\boldsymbol{\xi}$,
\begin{equation}
\label{eq:divgn}
\bnabla\bcdot\boldsymbol{g}^0
=-\alpha^2v_r^2\sum_{i=1}^{n}\sin\theta_i\,\sin2\theta_{i+1}\,\cos\theta_{i+2},
\end{equation}
and the pressure of Theorem~\ref{thm:main} is
\begin{equation}
\label{eq:pn}
p^0 = -\frac{\rho v_r^2}{6}\sum_{i=1}^{n}
\sin\theta_i\,\sin2\theta_{i+1}\,\cos\theta_{i+2}.
\end{equation}
\end{corollary}

\begin{proof}
Differentiating \eqref{eq:gclosed} with respect to $x_{\langle i\rangle}$ and writing
$S_i:=\sin\theta_i\sin2\theta_{i+1}\cos\theta_{i+2}$, the three terms contribute
$\tfrac12\alpha^2v_r^2S_{i-2}$, $-2\alpha^2v_r^2S_{i-1}$ and $\tfrac12\alpha^2v_r^2S_i$. Summing
over $i$ and relabelling each cyclic sum gives
$(\tfrac12-2+\tfrac12)\alpha^2v_r^2\sum_iS_i$, which is \eqref{eq:divgn}. Every harmonic of $S_i$
satisfies \eqref{eq:k6}, so $\Delta S_i=-6\alpha^2S_i$; hence \eqref{eq:pn} satisfies
$\Delta p^0=-\rho\,\bnabla\bcdot\boldsymbol{g}^0$, which is \eqref{eq:ppe} at $t=0$. It has zero
cell mean, so it is the solution fixed in Theorem~\ref{thm:main}.
\end{proof}

At $n=3$, \eqref{eq:pn} is the expression \eqref{eq:p3} already obtained. We can now also settle the
magnitudes.

\begin{proposition}[Norms]
\label{prop:norms}
For every $n\ge3$ and every phase vector,
\begin{equation}
\label{eq:scaling}
\lVert\boldsymbol{g}^0\rVert=\alpha v_r^2\,\frac{\sqrt{3n}}{4},
\qquad
\lVert\Uc^0\rVert=\alpha v_r^2\sqrt{\frac{n}{6}},
\qquad
\frac{\lVert\Uc^0\rVert}{\lVert\boldsymbol{g}^0\rVert}=\frac{2\sqrt2}{3}\approx0.9428 .
\end{equation}
In particular $\Uc^0\neq\boldsymbol{0}$ in every dimension and at every phase vector, so the forcing
of Theorem~\ref{thm:main} is never trivial.
\end{proposition}

\begin{proof}
The three terms of \eqref{eq:gclosed} are mutually orthogonal, as noted above, and within each term
the three coordinates involved are distinct, so each mean square factorises. Using
$\langle\sin^2\rangle=\langle\cos^2\rangle=\tfrac12$ on each factor,
\[
\bigl\langle (g_i^0)^2\bigr\rangle
=\alpha^2v_r^4\Bigl(\tfrac14\cdot\tfrac18+\tfrac18+\tfrac14\cdot\tfrac18\Bigr)
=\tfrac{3}{16}\,\alpha^2v_r^4 .
\]
The construction is invariant under the cyclic shift of the coordinates, so every component
contributes equally and
$\lVert\boldsymbol{g}^0\rVert^2=\tfrac{3}{16}n\,\alpha^2v_r^4$, which is the first formula.

For the second, the summands $S_i$ of \eqref{eq:pn} are mutually orthogonal, by the argument used
for \eqref{eq:supports}. Each $S_i$ is a product of three factors in the three distinct coordinates
$x_{\langle i\rangle}$, $x_{\langle i+1\rangle}$, $x_{\langle i+2\rangle}$, the doubled frequency
being carried by the middle one, so every harmonic of $S_i$ has
$\lvert k_{\langle i+1\rangle}\rvert=2\alpha$. Let $j\neq i$. Since $m\mapsto\langle m+1\rangle$ is a
bijection, $\langle j+1\rangle\neq\langle i+1\rangle$, and at the coordinate
$x_{\langle i+1\rangle}$ every harmonic of $S_j$ has
$\lvert k_{\langle i+1\rangle}\rvert$ equal to $\alpha$, if $\langle i+1\rangle$ is one of
$\langle j\rangle$ or $\langle j+2\rangle$, and to $0$ otherwise. Neither equals $2\alpha$, so the
supports of $S_i$ and $S_j$ are disjoint; as before this is independent of signs and of cyclic
coincidences, and holds at $n=3$ as at any other $n$. Each $S_i$ has mean square $\tfrac18$, whence
$\lVert p^0/\rho\rVert^2=\tfrac{1}{36}v_r^4\cdot\tfrac{n}{8}=nv_r^4/288$. Every harmonic satisfies
\eqref{eq:k6}, so
$\lVert\bnabla p^0/\rho\rVert^2=6\alpha^2\lVert p^0/\rho\rVert^2=n\alpha^2v_r^4/48$. The Helmholtz
decomposition is orthogonal in $L^2(\Omega)$ and
$\boldsymbol{g}^0=\Uc^0-\bnabla p^0/\rho$, so
\[
\lVert\Uc^0\rVert^2=\lVert\boldsymbol{g}^0\rVert^2-\lVert\bnabla p^0/\rho\rVert^2
=\alpha^2v_r^4\Bigl(\frac{3n}{16}-\frac{n}{48}\Bigr)=\alpha^2v_r^4\,\frac{n}{6},
\]
and the ratio follows.
\end{proof}

Proposition~\ref{prop:norms} holds for every $n\ge3$, so no table is needed to establish it.
Table~\ref{tab:dims} nonetheless records the agreement of \eqref{eq:scaling} with independent
computation in exact rational arithmetic, in units with $\alpha=v_r=1$, and is carried to $n=12$
simply to show that nothing distinguishes the dimensions beyond those reached in
\citet{tha2026}.

\begin{table}[ht]
\centering
\begin{tabular}{crrr}
\toprule
$n$ & $\lVert\boldsymbol{g}^0\rVert$ & $\lVert\Uc^0\rVert$
  & $\lVert\Uc^0\rVert/\lVert\boldsymbol{g}^0\rVert$ \\
\midrule
3 & $0.7500$ & $0.7071$ & $0.9428$ \\
4 & $0.8660$ & $0.8165$ & $0.9428$ \\
5 & $0.9682$ & $0.9129$ & $0.9428$ \\
6 & $1.0607$ & $1.0000$ & $0.9428$ \\
7 & $1.1456$ & $1.0801$ & $0.9428$ \\
8 & $1.2247$ & $1.1547$ & $0.9428$ \\
9 & $1.2990$ & $1.2247$ & $0.9428$ \\
10 & $1.3693$ & $1.2910$ & $0.9428$ \\
11 & $1.4361$ & $1.3540$ & $0.9428$ \\
12 & $1.5000$ & $1.4142$ & $0.9428$ \\
\bottomrule
\end{tabular}
\caption{The convective field of \eqref{eq:family} and its transverse part, in units with
$\alpha=v_r=1$, confirming Proposition~\ref{prop:norms}, which itself holds for every $n\ge3$. The
magnitudes grow as $\sqrt{n}$; the
ratio of transverse to total $L^2$ magnitude is independent of dimension and phase alike. In
squared norm the corresponding share is
$\lVert\Uc^0\rVert^2/\lVert\boldsymbol{g}^0\rVert^2=8/9$.}
\label{tab:dims}
\end{table}

The solution of Section~\ref{sec:main} is therefore explicit in space and time in every dimension.
Combining \eqref{eq:family}, \eqref{eq:pn} and \eqref{eq:solnfam}, and writing $i=1,\dots,n$ for the
velocity components,
\begin{equation}
\label{eq:solnn}
\begin{aligned}
v_i(\boldsymbol{x},t) &= v_r\bigl[
\sin\theta_i\sin\theta_{i-1}+\cos\theta_i\cos\theta_{i+1}
\bigr] e^{-2\alpha^2\kappa t},
\\[3pt]
p(\boldsymbol{x},t) &= -\frac{\rho v_r^2}{6}\sum_{i=1}^{n}
\sin\theta_i\,\sin2\theta_{i+1}\,\cos\theta_{i+2}\; e^{-4\alpha^2\kappa t},
\\[3pt]
f_i(\boldsymbol{x},t) &=
\Bigl[g_i^0+\frac{1}{\rho}\frac{\partial p^0}{\partial x_i}\Bigr] e^{-4\alpha^2\kappa t},
\qquad
g_i^0 \text{ as in \eqref{eq:gclosed}},
\end{aligned}
\end{equation}
with $E(t)=E(0)e^{-4\alpha^2\kappa t}$ throughout. At $n=3$ these reduce to \eqref{eq:soln3}, in
which the body force has in addition been reduced to the explicit trigonometric sum \eqref{eq:U3}.

Two consequences follow. For fixed $\alpha$ and $v_r$, the forcing norm grows only as $\sqrt{n}$.
And since the transverse-to-total ratio is the same constant in every dimension, the transverse
part of the convective field stands in the same proportion to the whole throughout: the benchmark
is of uniform character across $n$, which is what makes comparison across $n$ meaningful. Both
statements compare dimensions at fixed $\alpha$ and $v_r$, hence at a fixed spatial scale
$L=2\pi/\alpha$.

This stands in sharp contrast to \citet{tha2026}, where the admissible phase set is empty for
$5\le n\le8$ and the dimension-dependence is central to the result. In the present setting, every
phase vector in every dimension admits a solution, and the only dimensional dependence is the
factor $\sqrt{n}$.

The exact construction also provides a controlled way to ask a different question. Because the
transverse nonlinear term can be identified explicitly and balanced without adding net kinetic
energy, it can be selectively retained rather than cancelled everywhere. This makes it possible to
examine whether nonlinear activity can produce strong local velocity gradients, and hence vorticity
concentration, while separating that effect from direct energy input by the applied force.

\section{The forced solution as a reference for vorticity concentration}
\label{sec:conc}

In Theorem~\ref{thm:main} the body force cancels the transverse part of
the nonlinear term everywhere, so that the velocity evolves by diffusion
alone and, by \eqref{eq:work}, the forcing does no net work. That exact
solution provides a reference point for a different question: can the
same mechanism be modified so that the nonlinear term remains active in
a selected region of the domain, while the zero-work property is
retained?

Vorticity concentration refers here to the development of increasingly large vorticity within a
progressively more localised part of the flow. Since the vorticity
$\boldsymbol{\omega}=\bnabla\times\boldsymbol{v}$ is formed from spatial derivatives of the
velocity, growth of its maximum norm signals the formation of strong local velocity gradients. This
is distinct from growth of the total kinetic energy: the energy is a global measure of the velocity
field, whereas peak vorticity is sensitive to local structure. A flow may therefore develop strong
vorticity concentration even while its total kinetic energy is decreasing.

The motivation is to study that concentration without confusing it with direct injection of kinetic
energy by the applied force. One of the questions considered below is precisely whether vorticity
can become strongly concentrated while the total energy continues to decrease.

The computations in this section are intended as a controlled numerical
experiment rather than as evidence of singular behaviour. They also
illustrate an important numerical point: apparent saturation of peak
vorticity can be produced by insufficient spatial resolution.

\subsection{A localised force}

Choosing the force to be $\Uc$, the transverse part of the convective
field introduced in \eqref{eq:UF}, cancels the nonlinearity everywhere.
We now modify that construction so that the cancellation occurs only
outside a prescribed region.

Let $\chi(\boldsymbol{x})$ be a localisation function centred at the origin, which we take at the
centre of the periodic cell, and defined through two radii $r_1<r_2$. With $\alpha=\pi$ the cell
has side $2$, so that the outer ball lies strictly within it and no question of periodic wrapping
arises. It equals unity on the inner ball, decreases smoothly from one to zero across
the transition shell between the two radii, and vanishes outside the outer ball:
\begin{equation}
\label{eq:chi}
\chi(\boldsymbol{x})=
\begin{cases}
1, & s\le r_1,\\[2pt]
\tfrac12\Bigl[1+\cos\Bigl(\pi\dfrac{s-r_1}{r_2-r_1}\Bigr)\Bigr],
  & r_1<s<r_2,\\[6pt]
0, & s\ge r_2,
\end{cases}
\qquad s=\lvert\boldsymbol{x}\rvert .
\end{equation}
The raised-cosine taper is continuous together with its first derivative at both radii, so no
discontinuity is introduced at the edge of the active region and the applied force is continuously
differentiable. It is not of class $C^\infty$: the second derivative jumps at $r_1$ and $r_2$. That has a
consequence worth stating, since the applied force inherits the taper's smoothness and no more. The
Fourier coefficients of a $C^1$ field decay algebraically rather than faster than any power, so the
asymptotic spatial convergence of the localised computations is algebraic, not spectral. It does
not affect the comparisons made here, which are between computations on different grids using the
identical taper, but it does mean that smoothness cannot be invoked in place of the resolution
diagnostic of Section~\ref{sec:method}. A $C^\infty$ bump function would restore spectral
convergence at the cost of repeating every localised computation.
In the computations reported below $r_1=0.6$ and $r_2=0.75$, so that the transition shell has
thickness $0.15$; the single radius quoted in Table~\ref{tab:conc} is $r_1$.

The three regions have distinct roles, shown in Figure~\ref{fig:chi}. Within the inner ball the
transverse nonlinear term remains fully active. Across the transition shell its cancellation is
gradually switched on. Outside the outer ball it is cancelled entirely, as in
Theorem~\ref{thm:main}.

\begin{figure}[ht]
\centering
\includegraphics[width=\textwidth]{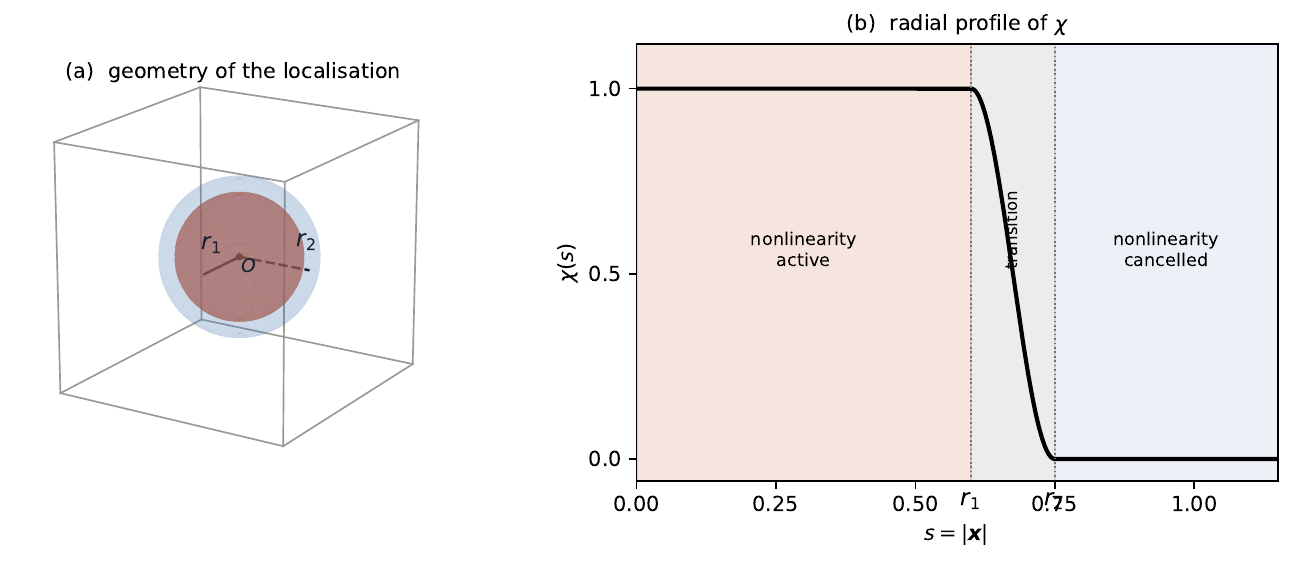}
\caption{Geometry of the localised forcing. (a)~The inner ball $\lvert\boldsymbol{x}\rvert\le r_1$,
centred at the origin within the periodic cell, is the region in which the transverse nonlinear
term remains fully active. Across the surrounding shell $r_1<\lvert\boldsymbol{x}\rvert<r_2$ the
localisation function decreases smoothly from one to zero. Outside $r_2$ the transverse nonlinear
term is cancelled by the applied force. (b)~Radial profile of $\chi$ from \eqref{eq:chi}, with
$r_1=0.6$ and $r_2=0.75$; the raised-cosine taper is continuous together with its first derivative
at both radii.}
\label{fig:chi}
\end{figure}

Consider first the factor $(1-\chi)\Uc$. Inside the inner ball, where
$\chi=1$, this factor vanishes, and hence the nonlinear term is left
unchanged. Outside the outer ball, where $\chi=0$, it becomes $\Uc$ and
cancels the transverse nonlinear term just as in
Theorem~\ref{thm:main}. This suggests the force
$(1-\chi)\Uc$.

There is, however, an important complication. Although the complete
field $\Uc$ satisfies
\[
\langle\Uc\bcdot\boldsymbol{v}\rangle=0,
\]
the truncated field $(1-\chi)\Uc$ does not in general satisfy the same
identity. Multiplication by the spatially varying function $1-\chi$
changes the structure of the field. Consequently, the truncated force
may perform non-zero work on the velocity and thereby inject or remove
kinetic energy.

To eliminate this unwanted energy transfer, we add a component parallel
to the velocity. We therefore define
\begin{equation}
\label{eq:floc}
\boldsymbol{f}
=
(1-\chi)\,\Uc-\beta(t)\,\boldsymbol{v},
\qquad
\beta(t)
=
\frac{\langle(1-\chi)\Uc\bcdot\boldsymbol{v}\rangle}
     {\langle\lvert\boldsymbol{v}\rvert^{2}\rangle}.
\end{equation}
Here $\langle\cdot\rangle$ denotes the spatial average over the periodic cell, as in
\eqref{eq:cell}.

One point must be made at once, because it governs everything that follows. Equation
\eqref{eq:floc} defines the \emph{applied} force. It is not divergence-free: multiplication by the
spatially varying factor $1-\chi$ destroys the solenoidality of $\Uc$, so
$\PT[(1-\chi)\Uc]\neq(1-\chi)\Uc$ in general. What acts on the velocity is therefore not
$\boldsymbol{f}$ but its projection $\PT\boldsymbol{f}$, and the two must be kept apart: the
evolution equation is written in terms of $\PT\boldsymbol{f}$ in \eqref{eq:locproj} below, and it is
$\PT\boldsymbol{f}$ that the computations impose. The zero-work property is unaffected by the
distinction, because $\PT$ is self-adjoint and $\PT\boldsymbol{v}=\boldsymbol{v}$ give
$\langle\PT\boldsymbol{f}\bcdot\boldsymbol{v}\rangle
=\langle\boldsymbol{f}\bcdot\boldsymbol{v}\rangle$; this is shown at \eqref{eq:workproj}.

The role of $\beta(t)$ can be seen directly. Taking the inner product
of \eqref{eq:floc} with $\boldsymbol{v}$ and averaging gives
\[
\langle\boldsymbol{f}\bcdot\boldsymbol{v}\rangle
=
\langle(1-\chi)\Uc\bcdot\boldsymbol{v}\rangle
-
\beta(t)\langle\lvert\boldsymbol{v}\rvert^2\rangle.
\]
Substituting the definition of $\beta(t)$ makes the two terms cancel,
so that
\begin{equation}
\label{eq:zerowork}
\langle\boldsymbol{f}\bcdot\boldsymbol{v}\rangle=0.
\end{equation}
Thus the force is constructed to do exactly zero net work at every
time.

The force \eqref{eq:floc} is not itself divergence-free: multiplication by $1-\chi$ destroys the
solenoidality of $\Uc$, so that in general $\PT[(1-\chi)\Uc]\neq(1-\chi)\Uc$. Only its projection
enters the velocity equation. Since $\PT\boldsymbol{v}=\boldsymbol{v}$, substitution into the
projected equation $\partial_t\boldsymbol{v}=\kappa\Delta\boldsymbol{v}-\Uc+\PT\boldsymbol{f}$ gives
\begin{equation}
\label{eq:locproj}
\frac{\partial\boldsymbol{v}}{\partial t}
=
\kappa\Delta\boldsymbol{v}
-\Uc+\PT\bigl[(1-\chi)\Uc\bigr]
-\beta(t)\boldsymbol{v}
=
\kappa\Delta\boldsymbol{v}
-\PT\bigl[\chi\,\Uc\bigr]
-\beta(t)\boldsymbol{v},
\end{equation}
the second form following from $\Uc=\PT\Uc$. It is \eqref{eq:locproj}, with the projection of the
localised field computed explicitly, that is integrated in the computations below. The
construction is transparent in this form. Before projection, the localisation leaves the transverse
nonlinear term active where $\chi\approx1$ and suppresses it where $\chi\approx0$; the Leray
projection then restores solenoidality. Since $\PT$ is non-local, the projected field does not have
compact support within the inner ball, and the reading above is of the field before projection
rather than after it. The scalar $\beta(t)$ maintains the zero-work condition.

That condition survives the projection. Because $\PT$ is self-adjoint and
$\PT\boldsymbol{v}=\boldsymbol{v}$,
\begin{equation}
\label{eq:workproj}
\langle\PT\boldsymbol{f}\bcdot\boldsymbol{v}\rangle
=\langle\boldsymbol{f}\bcdot\PT\boldsymbol{v}\rangle
=\langle\boldsymbol{f}\bcdot\boldsymbol{v}\rangle
=0,
\end{equation}
so the part of the force that actually acts on the velocity also does no net work, and applying
$\PT$ to the localised field introduces no energy of its own.

The limiting case $\chi\equiv1$ is particularly useful. Then
$1-\chi=0$, so that
\[
\beta(t)=0,
\qquad
\boldsymbol{f}=\boldsymbol{0},
\]
and the ordinary unforced Navier--Stokes equations are recovered.
Consequently, the localised calculation can be compared directly with
an unforced calculation using the same initial condition, Reynolds
number, numerical method, and spatial resolution.

The neutralising term $-\beta(t)\boldsymbol{v}$ is essential. The
orthogonality relation
\[
\langle\boldsymbol{v}\bcdot\Uc\rangle=0
\]
used in Theorem~\ref{thm:main} applies to the complete transverse
convective field. After multiplication by $1-\chi$, that structure is
lost, and there is no reason for
\[
\langle(1-\chi)\Uc\bcdot\boldsymbol{v}\rangle
\]
to vanish.

This effect is visible numerically. Without the neutralising term, the
measured work at $t=0$ for $r_1=0.6$, $r_2=0.75$ is
$+5.6\times10^{-3}$, compared with an energy of $0.11$. Over an interval
of length two, this unintended forcing accounts for an energy increase
of nearly eight per cent. Such an increase would make it difficult to
distinguish genuine concentration from concentration assisted simply by
energy injection.

With the correction in \eqref{eq:floc}, the situation changes
completely. In every corrected run reported below,
\[
\left|
\langle\boldsymbol{f}\bcdot\boldsymbol{v}\rangle
\right|
\le 2.2\times10^{-19},
\]
which is numerically indistinguishable from zero at the precision and the scales of the
computation. (It is well below double-precision machine epsilon, $\approx2.2\times10^{-16}$,
because the absolute scale of the quantity being cancelled is itself small; the figure should not
be read as machine epsilon.) The resulting experiment
therefore isolates redistribution and concentration from direct energy
input by the body force.

\subsection{Method, parameters and resolution control}
\label{sec:method}

\paragraph{Scheme.} The equations are solved pseudo-spectrally on the periodic cell, which for the
computations of this section is taken to be $\Omega=[0,2]^3$, so that the fundamental wavenumber is
$\pi$ and admissible wavevectors are $\boldsymbol{k}\in\pi\mathbb{Z}^3$. The nonlinear term is
evaluated in physical space and dealiased by the two-thirds rule. Time integration uses an explicit
midpoint (second-order Runge--Kutta) scheme together with an integrating factor
$e^{-\kappa\lvert\boldsymbol{k}\rvert^2\Delta t}$ for the viscous term, which is therefore treated
exactly.

The projection $\PT$ is exact on the represented modes, and the integrating factor is exact, so the
error comes from the truncation in space and from the treatment of the nonlinear term in time. For
the unforced computations ($\chi\equiv1$) the solution is smooth and the spatial convergence is
spectral, the error decreasing faster than any power of the grid spacing. For the localised
computations this is not so: the taper \eqref{eq:chi} is only $C^1$, so the applied force has
algebraically decaying Fourier coefficients and the asymptotic spatial convergence of those runs is
algebraic rather than spectral. This does not affect the comparisons made below, which are between
computations on different grids using the identical taper, but it is a further reason not to rely
on smoothness alone and to monitor resolution directly.

Whatever the scheme, the principal limitation is the finite number of retained modes, and no
reformulation removes it---a flow whose gradients steepen will eventually outrun any fixed
truncation. The purpose of the diagnostics below is accordingly not to compensate for a deficient
method but to determine, at each instant, whether the resolution in hand is sufficient for the
quantity being reported. That determination is part of the result: the first finding of
Section~\ref{sec:findings} is precisely that a computation run without such a check yields a peak
vorticity that is an artefact of the truncation rather than a property of the flow.

\paragraph{Reynolds number.} Lengths are scaled so that $\Omega=[0,2]^3$, and the initial field is
normalised so that its peak speed is unity. The Reynolds number is then
\begin{equation}
\label{eq:Re}
\mathcal{R}e=\frac{UL_\ast}{\kappa}=\frac{1}{\kappa},
\qquad U=\max_{\boldsymbol{x}}\lvert\boldsymbol{v}(\boldsymbol{x},0)\rvert=1,
\qquad L_\ast=1,
\end{equation}
the reference length $L_\ast$ being one half of the period. Equivalently, every computation reported
here sets $\kappa=1/\mathcal{R}e$ with the normalisations just stated.

\paragraph{Initial condition.} The initial field is not the analytical family \eqref{eq:family}; it
is a generic smooth solenoidal field, chosen so that the concentration study is not reporting a
property peculiar to the exact solution. It is constructed as follows. For each wavevector
$\boldsymbol{k}\in\pi\mathbb{Z}^3$ on the two shells
\begin{equation}
\label{eq:shells}
\lvert\boldsymbol{k}\rvert^2/\pi^2 = 2
\qquad\text{and}\qquad
\lvert\boldsymbol{k}\rvert^2/\pi^2 = 6,
\end{equation}
taken once per conjugate pair $\{\boldsymbol{k},-\boldsymbol{k}\}$, an orthonormal pair
$\boldsymbol{e}_1(\boldsymbol{k})$, $\boldsymbol{e}_2(\boldsymbol{k})$ spanning
$\boldsymbol{k}^{\perp}$ is constructed deterministically, and
\begin{equation}
\label{eq:ic}
\widehat{\boldsymbol{v}}(\boldsymbol{k},0)
= c_{\lvert\boldsymbol{k}\rvert}\bigl[(a_1+\mathrm{i}b_1)\boldsymbol{e}_1
+(a_2+\mathrm{i}b_2)\boldsymbol{e}_2\bigr],
\qquad
\widehat{\boldsymbol{v}}(-\boldsymbol{k},0)=\overline{\widehat{\boldsymbol{v}}(\boldsymbol{k},0)},
\end{equation}
with $a_1,b_1,a_2,b_2$ independent standard normal variates, $c=1$ on the first shell and $c=1/2$ on
the second. Placing the coefficients in $\boldsymbol{k}^\perp$ makes the field solenoidal by
construction, and the conjugate condition makes it real. The result is projected, dealiased, and
finally rescaled so that $\max_{\boldsymbol{x}}\lvert\boldsymbol{v}(\boldsymbol{x},0)\rvert=1$. The
variates are drawn from the Mersenne Twister generator seeded with the value $17$, and the basis
$\boldsymbol{e}_1,\boldsymbol{e}_2$ is built by explicit orthogonalisation rather than by an
eigenvalue routine, so that the same initial condition is obtained on different platforms. All
computations in this section, at every grid and every Reynolds number, start from this same field.
The same Fourier coefficients \eqref{eq:ic} are generated at every $N$, from the same seed and in
the same order, before any grid-dependent step is applied, so the spectral content of the initial
field is independent of the grid. The normalisation and the dealiasing mask are not,
so the realised initial energy
$E(0)=\tfrac12\langle\lvert\boldsymbol{v}\rvert^2\rangle$ varies in the fourth decimal place with
resolution, from $0.1098$ at $64^3$ to $0.1095$ at $128^3$. The initial peak vorticity is $6.28$ on
every grid used, and all amplification factors quoted below are referred to it.

\paragraph{Time step.} The time step is fixed within each run and set by
\begin{equation}
\label{eq:dt}
\Delta t = 0.006\,\frac{32}{N}\left(\frac{1000}{\mathcal{R}e}\right)^{1/2},
\end{equation}
$N$ being the number of grid points per direction; this gives $\Delta t=3.00\times10^{-3}$ at
$N=64$, $\mathcal{R}e=1000$, and $\Delta t=8.66\times10^{-4}$ at $N=128$, $\mathcal{R}e=3000$. The
rule is deliberately conservative, scaling with the grid spacing and with the viscous time. Since
the conclusions below turn on distinguishing numerical artefacts from physical behaviour, the
temporal error was checked directly by refinement; the result is reported in
Section~\ref{sec:findings}.

\paragraph{Spectral content and the tail diagnostic.} The initial field occupies only the two
shells \eqref{eq:shells}, so the computation begins with a smooth, spectrally compact velocity field
containing no initially populated high-wavenumber modes. As the flow evolves, nonlinear
interactions transfer activity to higher wavenumbers, corresponding to increasingly fine spatial
structure. This matters particularly for the vorticity, since taking a spatial derivative weights
each Fourier mode by its wavenumber: the maximum vorticity is therefore more sensitive than the
velocity itself to the development of poorly resolved small scales. A computation may look
perfectly smooth in physical space while already lacking the spectral resolution to represent its
largest gradients.

Resolution is therefore monitored throughout each run. Let
\begin{equation}
\label{eq:shellE}
E_s=\tfrac12\!\!\sum_{s-\frac12\,\le\,\lvert\boldsymbol{k}\rvert/\pi\,<\,s+\frac12}\!\!
\lvert\widehat{\boldsymbol{v}}(\boldsymbol{k})\rvert^2 ,
\qquad s=1,\dots,S,\quad S=N/2,
\end{equation}
be the kinetic energy in the unit-width radial shell centred on $\lvert\boldsymbol{k}\rvert=s\pi$.
The spectral tail is then
\begin{equation}
\label{eq:tail}
T=\frac{\displaystyle\sum_{s>0.85\,S}E_s}{\displaystyle\max_{1\le s\le S}E_s},
\end{equation}
the numerator running over the highest fifteen per cent of the radial shells and the denominator
being the largest shell energy, that is, the spectral peak. A small $T$ indicates that the
high-wavenumber end of the spectrum carries little energy and is still well separated from the
dominant scales; conversely, if appreciable energy reaches the highest wavenumbers available on the
grid, the computation is approaching the smallest scales it is capable of representing.

We use $T=10^{-4}$ as the level beyond which a computation is regarded as under-resolved. The value
is a working convention rather than a derived bound, and it is necessary but not sufficient: the
$96^3$ case below satisfies it and is nonetheless some twenty per cent away from its refined value.
For that reason we distinguish in what follows between a run that \emph{passes the spectral-tail
criterion} and one that is \emph{grid-converged}, reserving the latter for cases supported by
explicit refinement. The conclusions below rest on direct comparison between grids wherever such a
comparison is available, with the tail used only to locate the intervals over which a single
computation can be trusted.

\begin{table}[ht]
\centering
\begin{tabular}{rrlrrr}
\toprule
grid & $\mathcal{R}e$ & force &
peak $\lVert\boldsymbol{\omega}\rVert_\infty$ &
at $t$ & tail max \\
\midrule
$32^3$  & $200$  & zero-work localised, $r_1=0.6$ & $7.02$  & $1.24$ & $2.3\times10^{-6}$ \\
$32^3$  & $200$  & none, $\chi\equiv1$            & $7.31$  & $1.74$ & $2.0\times10^{-6}$ \\
$32^3$  & $400$  & zero-work localised, $r_1=0.6$ & $10.47$ & $1.82$ & $4.9\times10^{-5}$ \\
$32^3$  & $400$  & none, $\chi\equiv1$            & $12.19$ & $2.06$ & $1.1\times10^{-4}$ \\
$128^3$ & $3000$ & zero-work localised, $r_1=0.6$ & $79.17$ & $3.16$ & $8.1\times10^{-6}$ \\
$128^3$ & $3000$ & none, $\chi\equiv1$            & $81.98$ & $3.57$ & $2.4\times10^{-4}$ \\
\bottomrule
\end{tabular}
\caption{Zero-work computations: peak vorticity attained, from an initial value of $6.28$. Every
run in this table uses the corrected force \eqref{eq:floc}, for which
$\langle\boldsymbol{f}\bcdot\boldsymbol{v}\rangle=0$ and the energy decreases monotonically, or no
force at all. The tail maximum is the diagnostic \eqref{eq:tail}; entries for which it exceeds
$10^{-4}$ do not pass the criterion and are reported for completeness. These are the computations on
which the conclusions of Section~\ref{sec:findings} rest.}
\label{tab:conc}
\end{table}

\begin{table}[ht]
\centering
\begin{tabular}{rrrrr}
\toprule
grid & $\mathcal{R}e$ &
peak $\lVert\boldsymbol{\omega}\rVert_\infty$ &
at $t$ & tail max \\
\midrule
$32^3$  & $200$  & $9.52$  & $0.73$ & $1.7\times10^{-6}$ \\
$64^3$  & $1000$ & $36.99$ & $2.72$ & $1.5\times10^{-5}$ \\
$96^3$  & $3000$ & $74.40$ & $3.93$ & $8.2\times10^{-5}$ \\
$128^3$ & $3000$ & $89.69$ & $3.01$ & $1.1\times10^{-5}$ \\
\bottomrule
\end{tabular}
\caption{Diagnostic calculations, retained only for the grid-refinement comparison of
Section~\ref{sec:findings}. These used the uncorrected localised force $(1-\chi)\Uc$, which
performs work on the fluid and does not hold the energy budget fixed. They are therefore not
comparable with Table~\ref{tab:conc} and carry none of the physical conclusions; the amplification
figures quoted in the text are taken from Table~\ref{tab:conc} alone.}
\label{tab:diag}
\end{table}

\subsection{Findings}
\label{sec:findings}

The computations establish three observations, taken in turn below: that loss of resolution can
mimic physical saturation; that resolved peak-vorticity amplification strengthens over the
Reynolds-number range examined; and that the localised forcing has not been shown to enhance
concentration systematically relative to the unforced dynamics. A fourth remark concerns what the
construction contributes methodologically.

\paragraph{Grid refinement shows that apparent saturation can be numerical.}

A maximum in the computed $\lVert\boldsymbol{\omega}\rVert_\infty$ invites the reading that
viscosity has overcome the concentrating action of the nonlinearity. It is necessary to distinguish
an \emph{observed numerical turnover} from a \emph{demonstrated physical saturation}: the former is
a property of a particular discretisation, the latter of the flow. A turnover seen on one grid is
not sufficient evidence for the second.

At $32^3$ and $\mathcal{R}e=1000$ the peak vorticity rises to $15.9$ near $t\approx1.1$ and then
declines, so that on that grid the computation appears to saturate. At $64^3$ the same
configuration reaches $25.6$ by $t=2$ and is still increasing; continued, it attains $36.99$ before
turning over near $t\approx2.7$, with the spectral tail never exceeding $1.5\times10^{-5}$. The
turnover on the coarser grid is therefore attributable to insufficient resolution. This is not to say that
physical saturation does not occur---the $64^3$ computation does turn over while remaining
resolved---but that a decline in computed peak vorticity carries no weight unless accompanied by a
diagnostic showing the small scales are still resolved.

The same caution applies at $\mathcal{R}e=3000$, where the evidence is a direct grid comparison.
The $96^3$ computation reaches a spectral tail of $8.2\times10^{-5}$, below the stated threshold of
$10^{-4}$, so that on the spectral criterion alone it would be accepted. Refining to $128^3$
nevertheless raises the peak vorticity from $74.40$ to $89.69$, an increase of twenty per cent,
with the tail falling to $1.1\times10^{-5}$. A computation that passes the threshold narrowly can
therefore still be materially under-resolved, and the threshold is best treated as a necessary
rather than a sufficient condition. We note that both of these computations used the uncorrected
force and gained energy by the same $7.7$ per cent; that gain is the work done by the force and not
a symptom of resolution, and we do not use it as one. Both comparisons in this paragraph---$32^3$
against $64^3$ at $\mathcal{R}e=1000$, and $96^3$ against $128^3$ at $\mathcal{R}e=3000$---are
drawn from Table~\ref{tab:diag}. They are used here solely as statements about resolution, for
which the energy budget is irrelevant, and each compares like with like.

\paragraph{The temporal error is negligible at the time steps used.}

Since these conclusions turn on separating numerical artefacts from physical behaviour, the fixed
step \eqref{eq:dt} was tested by direct refinement. The zero-work localised computation at $64^3$
and $\mathcal{R}e=1000$ was repeated from the identical initial condition with $\Delta t$,
$\Delta t/2$ and $\Delta t/4$, all other parameters held fixed. The peak vorticity attained was
$32.6760$, $32.6766$ and $32.6765$ respectively, each at $t=2.85$, with the largest spectral tail
$1.0\times10^{-5}$ in all three. The relative change is $1.6\times10^{-5}$ on the first refinement
and $1.2\times10^{-6}$ on the second, and the terminal energy agrees to six significant figures.
The temporal discretisation error is therefore some four orders of magnitude smaller than the
twenty per cent spatial difference between $96^3$ and $128^3$ discussed above. The discrepancies
this section is concerned with are spatial, not temporal, and \eqref{eq:dt} is amply conservative.

\paragraph{Resolved concentration strengthens over the Reynolds-number range examined.}

Within the zero-work computations, for which \eqref{eq:zerowork} holds and the kinetic energy
decreases monotonically, the amplification of peak vorticity relative to its initial value of
$6.28$ is $1.12$ at $\mathcal{R}e=200$, $1.67$ at $\mathcal{R}e=400$ and $12.61$ at
$\mathcal{R}e=3000$, the last on the $128^3$ grid with a spectral tail of $8.1\times10^{-6}$.
Figure~\ref{fig:vort} shows that computation directly.

One qualification attaches to the last of these figures, and it follows from the caution of the
preceding paragraphs. The $\mathcal{R}e=3000$ zero-work computation passes the spectral-tail
criterion comfortably, but it has not itself been refined: the $96^3$--$128^3$ comparison that
showed the criterion to be insufficient was made with the diagnostic force, not this one. The
figure $12.61$ should therefore be read as the amplification attained on the $128^3$ grid, and as a
lower bound on the refined value rather than a grid-converged one. A computation on a finer grid
would settle it, and the script that performs it accompanies the paper.

\begin{figure}[ht]
\centering
\includegraphics[width=\textwidth]{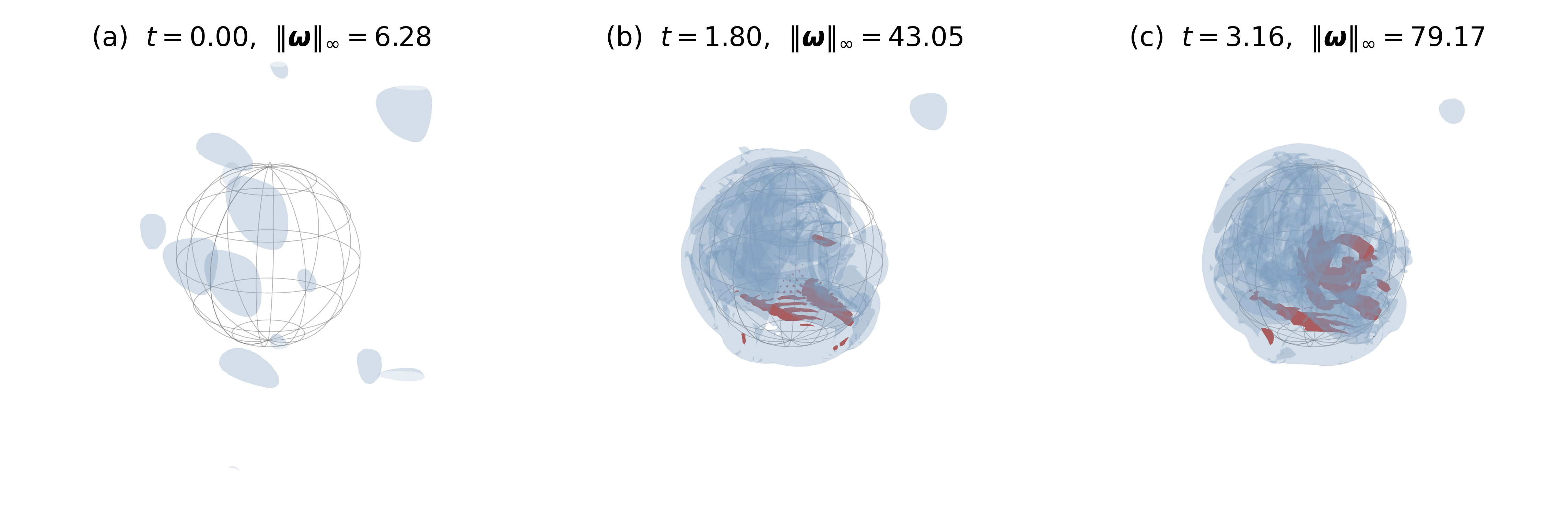}
\caption{Vorticity magnitude in the zero-work localised computation at $\mathcal{R}e=3000$ on the
$128^3$ grid, at three times. All three panels use the same camera, the same box and the same two
\emph{absolute} isosurface levels, $\lvert\boldsymbol{\omega}\rvert=5$ (outer, translucent) and
$\lvert\boldsymbol{\omega}\rvert=20$ (inner, opaque); no panel is normalised to its own maximum, so
the growth seen between them is the growth of the field. The higher level is absent at $t=0$,
where the peak is only $6.28$, and appears as compact structures at the later times. The faint
wireframe marks the sphere $\lvert\boldsymbol{x}\rvert=r_1$, within which the nonlinearity remains
active. The peak value of each panel is given above it; the field is drawn on the $128^3$ data
reduced to $64^3$ for plotting, the quoted maxima being those of the full-resolution field.}
\label{fig:vort}
\end{figure} The
time of the peak moves later as the Reynolds number rises. The corresponding figures for
Table~\ref{tab:diag}---$1.52$, $5.89$ and $14.28$---belong to the diagnostic computations, in which
the force was not corrected and the energy budget was not controlled; they are retained only for
the grid-refinement comparison and are not evidence for the present statement.

The distinction matters because the increase of $\lVert\boldsymbol{\omega}\rVert_\infty$ in the
zero-work series occurs while the total kinetic energy decreases. Stronger gradients are formed as
the energy falls, so the phenomenon is properly described as concentration rather than as energy
injection. We claim this for the Reynolds numbers and resolutions examined, not as a statement
about arbitrary $\mathcal{R}e$.

\paragraph{No systematic enhancement by localisation is established.}

The localised force was introduced to test whether leaving the nonlinearity active only within a
selected region produces stronger concentration than the unforced dynamics. The corrected,
zero-work computations do not establish such an enhancement. At $\mathcal{R}e=3000$ and $128^3$ the
localised run with $r_1=0.6$ reaches a reported maximum of $79.17$, against $81.98$ reported for
the unforced run. The same ordering of reported maxima is found at $\mathcal{R}e=200$ and $400$ on
the coarser grid, although at $\mathcal{R}e=400$ the unforced run has a spectral tail of
$1.1\times10^{-4}$, marginally beyond the threshold, so that comparison is not fully resolved
either.

An earlier apparent advantage for localisation arose from the uncorrected force, which performed
positive work, so that the two cases were not compared at equal energy budgets. Once
\eqref{eq:zerowork} is imposed, that twenty-per-cent advantage disappears.

The $\mathcal{R}e=3000$ comparison is more interesting than this summary suggests, and the
distinction between a \emph{reported} maximum and a \emph{resolved} maximum is essential to it. The
unforced computation develops a spectral tail of $2.4\times10^{-4}$ from about $t=3.2$, so its
reported maximum at $t=3.57$ lies in an under-resolved interval and is not a resolved maximum at
all. Comparing instead at $t=3.0$, where both computations satisfy the resolution criterion, the
localised run has $78.3$ against $70.3$ for the unforced run; the curves cross thereafter. During
the interval in which both are resolved, therefore, localisation produces the greater
concentration. Which evolution attains the larger \emph{resolved} maximum at $\mathcal{R}e=3000$
cannot be determined from the present computations. The results therefore do not establish a
systematic enhancement of concentration by localisation.

\paragraph{The principal value of the construction here is methodological.}

The resolved computations reported here provide evidence of strong vorticity concentration, but not
of singular behaviour; a finite computation cannot establish the absence of singular behaviour
beyond the interval over which it is resolved. In the zero-work cases the peak vorticity rises by a
factor of at least $12.61$ at $\mathcal{R}e=3000$, and where the subsequent decline occurs while the
computation remains adequately resolved, as at $\mathcal{R}e=1000$ on the $64^3$ grid, it is a
decline of the flow and not of the discretisation.

What Theorem~\ref{thm:main} supplies in this setting is the ability to manipulate the transverse
convective field while controlling the energy transferred by the applied force exactly. The
identity \eqref{eq:zerowork} separates changes produced by nonlinear redistribution from changes
produced by external energy input, and it is this separation, rather than any particular outcome of
the localisation experiment, that the construction contributes.

Table~\ref{tab:conc} may therefore be read as a reference set for further concentration studies.
An alternative forcing or localisation scheme can be compared against the same quantities: peak
vorticity, the time at which it occurs, the evolution of kinetic energy, the applied work, and the
high-wavenumber resolution diagnostic.

\section{Use as a benchmark}
\label{sec:bench}

The fields of Theorem~\ref{thm:main} provide a benchmark for numerical solvers of the incompressible
Navier--Stokes equations, and particularly for solvers in higher dimensions. They have four
properties that make them useful for this purpose. Velocity, pressure and body force are all known
in closed form. The time dependence is a pure exponential, known exactly. The energy decay is known
exactly, because the forcing does no net work. And the construction is available in every dimension
$n\ge3$, with no upper bound. The last is the principal motivation: solvers intended for $n>3$ have
few exact solutions against which to be tested, and those of \citet{tha2026} exist, over the range
that paper classifies, only at $n=3$ and $4$.

Two of these properties belong to Theorem~\ref{thm:main} itself and hold for any admissible field:
the exponential time dependence and the exact energy decay. The closed forms for the pressure and
the body force, and the numerical value of the transverse ratio quoted below, belong to the
particular family \eqref{eq:family}, for which Section~\ref{sec:dims} establishes them. A solver
tested against this benchmark is therefore tested against that family; the first two diagnostics
would apply to any other eigenfunction put through the same construction.

The zero-work property deserves particular emphasis. Because the energy must follow
$E(0)e^{-2\lambda\kappa t}$ exactly, a solver cannot reproduce the correct velocity while quietly
introducing spurious energy through its implementation of the forcing; any such error appears
directly in the energy.

Four diagnostics are natural, and they test different parts of a solver:
\begin{itemize}
\item the solution error $\lVert\boldsymbol{v}_{\mathrm{num}}-\boldsymbol{v}_{\mathrm{exact}}\rVert$,
in $L^2$ and, if wanted, $L^\infty$, which tests the solver as a whole;
\item the divergence $\lVert\bnabla\bcdot\boldsymbol{v}_{\mathrm{num}}\rVert$, which tests the
enforcement of incompressibility;
\item the energy $E(t)$ against $E(0)e^{-2\lambda\kappa t}$, which tests the viscous and temporal
evolution together with the forcing;
\item the ratio $\lVert\PT\boldsymbol{g}\rVert/\lVert\boldsymbol{g}\rVert$ against $2\sqrt2/3$, which
tests the evaluation of the nonlinear term and the projection.
\end{itemize}
The pressure is also known exactly, so a pressure solver can be tested independently, complementing
the more indirect test that the transverse ratio provides.

Mathematically, $\bnabla\bcdot\boldsymbol{v}=0$, $E(t)=E(0)e^{-2\lambda\kappa t}$ and the ratio
equals $2\sqrt2/3$ exactly. Numerically, the relevant question is how quickly the computed values
approach these as the discretisation is refined. A spectral projection method will typically hold
the divergence near machine precision, but other legitimate discretisations enforce
incompressibility differently, and for those the divergence should instead converge to zero at the
rate the method implies. The transverse ratio is especially useful here: unlike a residual that
need only become small, it is a specific non-zero target, independent of dimension and phase, that
a projection algorithm must reproduce. Because the solution is exact, computations on successively
finer grids or time steps also allow the observed order of convergence to be measured directly.

A solver validated in this way is on firmer ground when applied to flows in which gradients steepen,
such as those of Section~\ref{sec:conc}, where the resolution diagnostics become essential. The
same exact solution can therefore test incompressibility, viscous decay, nonlinear evaluation,
pressure projection, forcing implementation and convergence, across dimension. An exact benchmark
is an unforgiving one: with the velocity, pressure, forcing, energy decay and transverse ratio
all prescribed, a numerical error has relatively few places to hide.

\section{Reproducibility}
\label{sec:repro}

Three scripts accompany the paper, and they serve different purposes.

The first symbolically verifies the analytical results of Sections~\ref{sec:family}
to~\ref{sec:dims}: the solenoidality and eigenfunction property of \eqref{eq:family}, the
three-dimensional pressure \eqref{eq:p3} and body force \eqref{eq:U3}, the closed forms
\eqref{eq:scaling} together with their independence of the phases, and the residual of the exact
forced solution. These checks are independent confirmations; the lemmas and the theorem are
established by the proofs given in the text, not by the computations. The residual is obtained by
substituting the exact $\boldsymbol{v}$ and $\boldsymbol{f}$ of \eqref{eq:soln} into the projected
equation \eqref{eq:proj} and measuring what remains in the cell-averaged $L^2$ norm; the convective
term is recomputed from $\boldsymbol{v}(t)$ at each time rather than rescaled, so the test is not
circular. The residual does not exceed $1.3\times10^{-16}$, at the level of floating-point
round-off, at $t=0$, $0.5$, $1$, $2$ and $5$, for the field \eqref{eq:v03} and for members of
\eqref{eq:family} at $n=3$ and $n=4$ with randomly drawn phases; the applied work vanishes
identically in every case. The closed form \eqref{eq:gclosed} of Lemma~\ref{lem:gclosed} and the
divergence identity \eqref{eq:divgn} are confirmed for $3\le n\le8$ with $\alpha$, $v_r$ and
$\boldsymbol{\xi}$ symbolic; the norms \eqref{eq:scaling} and the pressure \eqref{eq:pn} are checked
for $3\le n\le12$ in exact rational arithmetic, at zero, symmetric and randomly drawn phase
vectors. The Fourier-support disjointness on which Lemma~\ref{lem:gclosed} and
Proposition~\ref{prop:norms} rest---between the three terms of \eqref{eq:gclosed}, and between the
$n$ summands of \eqref{eq:pn}---is checked by listing the supports explicitly for $3\le n\le12$, so
that the sign choices and the cyclic index coincidences, which are at their most severe at $n=3$,
are enumerated rather than assumed. The script also confirms the translation identity \eqref{eq:shift} for $3\le n\le8$, and
verifies for $3\le n\le6$---with
$\alpha$, $v_r$, $\rho$, $\kappa$ and $\boldsymbol{\xi}$ all carried as free symbols---that the
space-time solution \eqref{eq:solnn} satisfies the momentum equation of \eqref{eq:ns} identically,
with $\bnabla\bcdot\boldsymbol{v}=\bnabla\bcdot\boldsymbol{f}=0$.

The second script reproduces the numerical computations summarised in Table~\ref{tab:conc}. It
implements the scheme of Section~\ref{sec:conc}, including the localised force \eqref{eq:floc} with
its neutralising term, and reports at every output step the kinetic energy, the enstrophy, the peak
vorticity, the spectral-tail resolution diagnostic and the applied work
$\langle\boldsymbol{f}\bcdot\boldsymbol{v}\rangle$. Since the conclusions of that section depend on
distinguishing resolved from under-resolved intervals, the resolution diagnostic is part of the
result rather than an implementation detail. The initial field is generated deterministically from
a fixed seed, as set out in Section~\ref{sec:method}, so that the same initial condition can be
reproduced across platforms; the orthonormal basis of $\boldsymbol{k}^{\perp}$ in \eqref{eq:ic} is
built by explicit orthogonalisation rather than by an eigenvalue routine for the same reason. The
time step follows \eqref{eq:dt} and is not adaptive.

A third script, \texttt{refine160.py}, repeats the $\mathcal{R}e=3000$ zero-work computation on a
finer grid from the identical initial condition, so that the $128^3$ peak of Table~\ref{tab:conc}
can be checked against a refined value.

The scripts are named \texttt{verify\_force.py}, \texttt{concentration.py} and
\texttt{refine160.py}. They require Python~3.8 or later, with NumPy for the second and third and
SymPy in addition for the first, and read no data files. The run times of the higher-resolution computations in
Table~\ref{tab:conc} depend strongly on the hardware used; the $128^3$ computations are the
expensive ones.

\section{Concluding remarks}
\label{sec:concl}

For any smooth solenoidal periodic field that is a Laplacian eigenfunction, the transverse part of
its convective field, decaying at twice the viscous rate, is a body force for which that field is an
exact solution of the forced Navier--Stokes equations. The mechanism is a clean division of labour:
the pressure absorbs the longitudinal part of the convective field, the body force balances the
transverse part, and the forcing does no net work.

Applied to the $n$-parameter family \eqref{eq:family}, this yields exact forced solutions for every
phase vector and every dimension $n\ge3$, with velocity, pressure and force in closed form. There
is no upper bound on $n$: the pressure, the magnitudes and the non-vanishing of the forcing are
proved for all dimensions at once, not established case by case, because the construction's finite
interaction range makes each component the same function of five coordinates however large $n$ is.
The
forcing has magnitude $\lVert\Uc^0\rVert=\alpha v_r^2\sqrt{n/6}$, which for fixed $\alpha$ and
$v_r$ grows only as $\sqrt{n}$, and stands to the whole convective field in the fixed ratio of
$L^2$ magnitudes $2\sqrt2/3$, independently of $n$, of $\boldsymbol{\xi}$ and of the units. These fields
therefore supply an exact benchmark for incompressible solvers, and particularly for those intended
for $n>3$, where exact solutions are otherwise scarce: with velocity, pressure, force, energy decay,
incompressibility and the transverse ratio all known exactly, a method can be checked against
several independent quantities rather than one.

Used to modify rather than balance the convective field, the same apparatus gives a controlled
setting in which concentration can be studied with the energy budget held exactly. The zero-work
computations of Section~\ref{sec:conc} reach an amplification of at least $12.61$ at
$\mathcal{R}e=3000$ and
show that an apparent saturation of that growth can shift or disappear under refinement. They do
not establish a systematic enhancement of concentration by localising the nonlinearity.

The construction inverts the search that produced \citet{tha2026}. There the phase space is examined
for the rare assignments at which the transverse part vanishes. Here it is retained, and every phase
assignment in every dimension is productive. What obstructs the unforced problem is what solves the
forced one.

\section*{Ethics}
This work did not involve human participants, animal subjects, or fieldwork, and no ethical
approval was required.

\section*{Data accessibility}
No experimental or measured datasets were generated or analysed. All analytical results are
reproducible from the symbolic verification script, and all computational results from the
solver described in Section~\ref{sec:repro}; all three scripts are supplied as electronic
supplementary material.

\section*{Conflict of interest declaration}
The author declares no competing interests.

\section*{Funding}
No funding was received for this study.

\end{document}